\documentclass[journal,draftcls,onecolumn,12pt,twoside]{IEEEtran}
\IEEEoverridecommandlockouts
\usepackage{epsfig,graphics,amssymb}
\usepackage{amscd}
\usepackage{amsmath}
\usepackage{enumerate}
\usepackage{theorem}
\usepackage{cite}
\usepackage{stfloats}
\usepackage{epsfig}
\usepackage{verbatim}
\usepackage{color}
\theoremstyle{plain} \theorembodyfont{\itshape}
\theoremheaderfont{\scshape}
\newtheorem{theorem}{Theorem}
\theorembodyfont{\itshape} \theoremheaderfont{\scshape}
\newtheorem{lemma}{Lemma}

\theoremstyle{plain} \theorembodyfont{\itshape}
\theoremheaderfont{\scshape}

\usepackage{theoremref}
\newcommand{\eqdef}{\overset{\mathrm{\Delta}}{=}}

\begin{document}
\title{On Optimum Asymptotic Multiuser Efficiency of Randomly Spread CDMA}
\author{\IEEEauthorblockN{M. A. Sedaghat, R. R. M\"{u}ller, \emph{Senior Member, IEEE}, and F. Marvasti, \emph{Senior Member, IEEE}}
\thanks{
M. A. Sedaghat is with the Department of Electronics and Telecommunications, the Norwegian University of Science and Technology (NTNU), Trondheim, Norway (e-mail: mohammad.sedaghat@iet.ntnu.no).}
\thanks{R. R. M\"{u}ller is with Friedrich-Alexander Universit\"at Erlangen-N\"urnberg, Erlangen, Germany, and
The Norwegian University of Science and Technology, Trondheim, Norway, (e-mail: mueller@lnt.de).}
\thanks{F. Marvasti is with Department of Electrical Engineering, Sharif University of Technology, Tehran, Iran (email: marvasti@sharif.edu).}}

\maketitle

\begin{abstract}
We extend the result by Tse and Verd\'{u} on the optimum asymptotic multiuser efficiency of randomly spread CDMA with Binary Phase Shift Keying (BPSK) input. Random Gaussian and random binary antipodal spreading are considered. We obtain the optimum asymptotic multiuser efficiency of a $K$-user system with spreading gain $N$ when $K$ and $N\rightarrow\infty$ and the loading factor, $\frac{K}{N}$, grows logarithmically with $K$ under some conditions. It is shown that the optimum detector in a Gaussian randomly spread CDMA system has a performance close to the single user system at high Signal to Noise Ratio (SNR) when $K$ and $N\rightarrow\infty$ and the loading factor, $\frac{K}{N}$, is kept less than $\frac{\log_3K}{2}$. Random binary antipodal matrices are also studied and a lower bound for the optimum asymptotic multiuser efficiency is obtained. Furthermore, we investigate the connection between detecting matrices in the coin weighing problem and optimum asymptotic multiuser efficiency. We obtain a condition such that for any binary input, an $N\times K$ random matrix whose entries are chosen randomly from a finite set, is a detecting matrix as $K$ and $N\rightarrow \infty$.
\end{abstract}

\newcommand{\fzv}[2]{\noindent #1 \hfill \parbox{13.2cm}{#2}}
\def\mathlette#1#2{{\mathchoice{\mbox{#1$\displaystyle #2$}}%
                              {\mbox{#1$\textstyle #2$}}%
                              {\mbox{#1$\scriptstyle #2$}}%
                              {\mbox{#1$\scriptscriptstyle #2$}}}}
\newcommand{\matr}[1]{\mathlette{\boldmath}{#1}}
\newcommand{\SINR}{{sig\-nal--to--dis\-tor\-tion ratio}}
\newcommand{\SNR}{{sig\-nal--to--noise ratio}}
\renewcommand{\j}{{\rm j}}
\newcommand{\RR}{\mathbb{R}}
\newcommand{\CC}{\mathbb{C}}
\newcommand{\NN}{\mathbb{N}}
\newcommand{\ZZ}{\mathbb{Z}}
\newcommand{\pic}{\pi}
\newcommand{\deltaf}{\delta}
\def\argmin{\mathop{\rm argmin}}
\def\argmax{\mathop{\rm argmax}}
\newcommand{\diag}{{\rm diag}}
\def\expect{\mathop{\mbox{$\mathsf{E}$}}}
\newcommand{\hermite}[1]{{#1^{\rm H}}}
\newcommand{\transp}[1]{{#1^{\rm T}}}
\newcommand{\conj}[1]{{#1^{\ast}}}
\newcommand{\e}{{\rm e}}
\newcommand{\iu}{{\rm j}}
\newcommand{\vnull}{{\rm \bf 0}}
\newcommand{\I}{{\rm \bf I}}
\newcommand{\energy}{{E}}
\newcommand{\prob}[2]{{\rm p}_{#1}\!\!\left( #2 \right) }
\newcommand{\Prob}[2]{{\rm P}_{#1}\!\!\left( #2 \right) }
\newcommand{\proba}[2]{\breve{\rm p}_{#1}\!\!\left( #2 \right) }
\newcommand{\Proba}[2]{\breve{\rm P}_{#1}\!\!\left( #2 \right) }
\newcommand{\Probi}[2]{{\rm P}_{#1}^{-1}\!\left( #2 \right) }
\newcommand{\RT}[2]{{\rm R}_{#1}\!\left( #2 \right) }
\newcommand{\Q}{{\rm Q}}
\newcommand{\load}{{\beta}}
\newcommand{\sign}{{\rm sign}}
\newcommand{\dirac}[1]{\deltaf \! \left( #1 \right)}
\newcommand{\kron}[1]{\deltaf \! \left[ #1 \right]}
\newcommand{\MomGen}[2]{{\Phif}_{#1}\left( #2 \right) }
\newcommand{\define}{\stackrel{\triangle}{=}}
\newcommand{\D}{\displaystyle}
\newcommand{\eq}[1]{(\ref{#1})}
\newcommand{\eqs}[2]{(\ref{#1}) and (\ref{#2})}
\newcommand{\eqd}[3]{(\ref{#1}), (\ref{#2}), and (\ref{#3})}
\newcommand{\eqv}[4]{(\ref{#1}), (\ref{#2}), (\ref{#3}), and (\ref{#4})}
\newcommand{\alphabet}{{\cal B}}

\section{Introduction}

Calculating optimum asymptotic multiuser efficiency of CDMA systems for different signature codes has received attention because this parameter shows the performance loss of the optimum detector in comparison with the single user system when the background noise vanishes \cite{verdu1986optimum}-\cite{tse2000optimum}. Tse and Verd\'{u} in \cite{tse2000optimum} prove that the optimum asymptotic multiuser efficiency for a CDMA system with Binary Phase Shift Keying (BPSK) input signal and independent and identically distributed (i.i.d.) random spreading with $N$ chips approaches 1 as the number of users, $K$, tends to infinity and the loading factor, $\frac{K}{N}$, is kept equal to an arbitrary nonzero constant. As a consequence, in large scale randomly spread CDMA systems, the performance is close to the performance of the single user system at high Signal to Noise Ratio (SNR). In the related context of compressive sensing \cite{donoho2006compressed}, the authors of \cite{wu2010renyi} showed that the result of Tse and Verd\'{u} is not restricted to binary input signals, but holds for any input alphabet with finite cardinality. However, authors in \cite{tse2000optimum} and \cite{wu2010renyi} obtain their results only when the loading factor is kept finite and constant. In this paper, we generalize the theorem in \cite{tse2000optimum} and obtain the optimum asymptotic multiuser efficiency of a CDMA system with random spreading in a more general condition. Random binary antipodal and random Gaussian spreading matrices are considered. Moreover, the input signal is assumed to be BPSK. It is shown that the optimum asymptotic multiuser efficiency converges to 1 also when the loading factor grows logarithmically with the number of users. 

We also investigate the connection between detecting matrices in mathematics and the optimum asymptotic multiuser efficiency. Detecting matrices originated from the coin weighing problem in mathematics \cite{soderberg1963combinatory}-\cite{mow2009recursive}. We use the bounds in random detecting matrices as supplement of the obtained results on the optimum asymptotic multiuser efficiency. We also generalize the bound on random binary detecting matrices by Erd\H{o}s and R\'{e}nyi in \cite{erdgs1963two} for matrices whose entries are chosen randomly from a finite set of numbers. 


The rest of this paper is organized as follows. In section II, the main theorems about optimum asymptotic multiuser efficiency for binary antipodal and Gaussian random spreading are presented. In Section III, the connection between detecting matrices and optimum asymptotic multiuser efficiency is investigated. Finally, section IV concludes the paper.

\section{Optimum asymptotic multiuser efficiency}
Multiuser efficiency has been used as a common performance measure of detectors in CDMA systems. For the optimum detector, it is called optimum multiuser efficiency. Optimum \textit{asymptotic} multiuser efficiency was introduced by Verd\'{u} in \cite{verdu1986optimum} to measure the performance of the optimum receiver when noise vanishes. 

Assume a randomly spread CDMA system with a discrete model 
\begin{equation} \label{model}
\matr y=\matr{Hb}+\matr n,
\end{equation}
where $\matr{H}$ is an $N\times K $ spreading matrix whose elements are i.i.d. and have a symmetric probability density function (pdf)\footnote{A random variable $x$ has a symmetric pdf $\rho(x)$ if for every $\alpha$, $\rho(\alpha)=\rho(-\alpha)$}. $\matr{b}$ is the data vector that $b_i\in\{\pm 1\}$, $\matr{n}\sim\mathcal{N}(\matr{0},\sigma^2 \matr{I})$ is the additive white Gaussian noise vector and $\matr{y}$ is the received vector. Note that in \eqref{model}, the number of users is $K$ and the number of chips is assumed to be $N$. Moreover, the users are assumed to have unit power. In the considered model, all users have the same asymptotic multiuser efficiency \cite{verdu1998multiuser}
\begin{equation}\label{5}
\eta\eqdef 2 \lim_{\sigma\rightarrow 0} \sigma^2 \log\left( \frac{1}{{\rm P_e}(\sigma)}\right),
\end{equation}
where ${\rm P_e}(\sigma)$ is the bit-error rate of the users.
Then, the optimum asymptotic multiuser efficiency is calculated as follows \cite{verdu1998multiuser}
\begin{equation}\label{5}
\eta= \min_{\matr{x}\in \{\pm1,0\}^K\setminus\{\matr{{0}}\}} {\matr{x}^T\matr{R}\matr{x}} ,
\end{equation}
where $\matr{R}\eqdef \matr{H}^\dagger \matr{H}$ and $\matr{x}$ is the error vector.
$\eta$ is in $[0,1]$ for any given $K$ and $N$. In \cite{tse2000optimum}, it is proven that when $K,N\rightarrow \infty$ and $K/N$ is kept constant and finite, $\eta$ converges to $1$ almost surely. Therefore, an interesting question is that whether it is necessary to keep $K/N$ finite. In fact, the question is what is the maximum possible $K/N$ to have $\eta$ converging to 1. This question applies to compressive sensing as well. In compressive sensing, it is desired to find a transfer matrix with minimum number of rows to compress a sparse data vector \cite{donoho2006compressed}. Note that unlike compressive sensing in this paper we consider BPSK input signal which results in a ternary error vector. 
In the remaining parts of this section we obtain some sufficient bound on $K/N$ as an extension of the result in \cite{tse2000optimum}.

Let $E_K$ be the event that ${\matr{x}^T\matr{R}\matr{x}} <1$ for at least one nonzero error vector $\matr{x}\in\{\pm1,0\}^K$. Therefore, 
\begin{equation}\label{EK}
{\rm{P}}(E_K)={\rm{P}}\left(\bigcup_{\matr{x}\in \{\pm1,0\}^K\setminus\{\matr{{0}}\}} \matr{x}^T\matr{R}\matr{x}<1\right).
\end{equation}
By applying the union bound to \eqref{EK}, an upper bound is obtained as
\begin{equation}\label{uniobound}
{\rm{P}}(E_K)\leq \sum_{\matr{x}\in \{\pm1,0\}^K\setminus\{\matr{{0}}\}}{{\rm{P}}(\matr{x}^T\matr{R}\matr{x}<1)}.
\end{equation}
%
In the following parts, we consider random binary antipodal and random Gaussian spreading matrices. We derive some conditions that $\eta$, as defined in \eqref{5}, converges to 1 almost surely.  

\subsection{The optimum asymptotic multiuser efficiency for i.i.d. binary antipodal random spreading}
In this part, it is assumed that the entries of the spreading matrix, ${H}_{i,j}$, are chosen randomly from $\left\{\pm\frac{1}{\sqrt{N}}\right\}$ with equal probability. The input signal, $\matr{b}$, is also considered to be BPSK. We first present L\textsc{emma}~\ref{Lemma 1} , {L\textsc{emma}}~\ref{Lemma 2} and {L\textsc{emma}}~\ref{Lemma 3} which are used in the main theorems.\\
\begin{lemma}\label{Lemma 1}
Let $H_{i,j}\in\left\{\pm\frac{1}{\sqrt{N}}\right\}$. For every error vector $\matr{x}\in\{\pm1,0 \}^K$ with odd weight\footnote{By the weight of a vector we mean the number of nonzero elements of it.}, ${\rm{P}}\left(\matr{x}^TR\matr{x}<1\right)=0$.
\end{lemma}
\begin{proof}
 Let $\matr{u}(\matr{x})=\left[ u_1(\matr{x}),\cdots,u_N(\matr{x})\right]^T=\matr{H}\matr{x}$. Hence, we have
\begin{equation}\label{1}
\matr{x}^T\matr{R}\matr{x}=\parallel\matr u(\matr x)\parallel^2=\sum_{\ell=1}^{N}{u_\ell(\matr{x})^2}.
\end{equation}
For every $\matr{x}$ with an odd weight, it can be shown that
\begin{equation}\label{2}
u_\ell(\matr{x})^2\geq\frac{1}{{N}},
\end{equation}
therefore, \eqref{1} and \eqref{2} result in $\matr{x}^T\matr{R}\matr{x}\geq 1$ and this proves the lemma.
\end{proof}
\begin{lemma}\label{Lemma 2}
 Let $\matr{x}_j$ be an error vector with even weight $2j>0$ and $B_j$ be the event that the number of nonzero elements of $\matr{u}(\matr{x}_j)$ is less than $\frac{N}{4}$. Then, ${\rm{P}}(\matr{x}_j^T\matr{R}\matr{x}_j<1)\leq {\rm{P}}(B_j)$. 
\end{lemma}
\begin{proof}
 The weight of $\matr{x}_j$ is an even number. Therefore, for a nonzero element $u_\ell(\matr{x}_j)$ we have
\begin{equation}\label{3A}
u_\ell^2(\matr{x}_j)\geq \frac{4}{N}.
\end{equation}
If $\matr{x}_j^T\matr{R}\matr{x}_j<1$ then from \eqref{1} and \eqref{3A} it can be proven that $B_j$ happens. Therefore,
\begin{equation}\label{4A}
{\rm{P}}(B_j|\matr{x}_j^T\matr{R}\matr{x}_j<1)=1.
\end{equation}
From \eqref{4A}, it is concluded that 
\begin{equation}\label{4B}
 {\rm{P}}(B_j)\geq {\rm{P}}(\matr{x}_j^T\matr{R}\matr{x}_j<1),
\end{equation}
and this proves the lemma. 
\end{proof}
\begin{lemma}\label{Lemma 3}
 If the entries of $\matr{H}$ are i.i.d. and have a symmetric pdf $f(\cdot)$ then for all $\matr{x}\in\{\pm 1,0\}^K$ with a same weight, ${\rm{P}}(\matr{x}^T\matr{R}\matr{x}<1)$ is equal.
\end{lemma}
\begin{proof}
Proof is given in Appendix A.
\end{proof}

In the following theorem, a new bound for the optimum asymptotic multiuser efficiency of a CDMA system with binary antipodal random spreading matrix is presented.\\
\begin{theorem}\label{theorem1}
 For the CDMA system \eqref{model} with $\matr{b}\in\{\pm 1\}^K$ and $H_{i,j}\in\left\{\pm\frac{1}{\sqrt{N}}\right\}$, the optimum asymptotic multiuser efficiency converges to $1$ almost surely as $K,N\rightarrow \infty$, and $\frac{K}{N\log_3{K}}$ is kept less than $\frac{3}{8}$.
\end{theorem}
\begin{proof}
Based on {L\textsc{emma}}~\ref{Lemma 1}, \eqref{uniobound} can be written as
\begin{equation}\label{boundevenbinary}
{\rm{P}}(E_K)\leq \sum_{\matr{x}\in \{\pm1,0\}^K \setminus\{\matr{{0}}\},~ \text{even weight }\matr{x}}{{\rm{P}}(\matr{x}^T\matr{R}\matr{x}<1)}.
\end{equation}
The entries of $H$ are i.i.d.  and have a symmetric pdf therefore based on {L\textsc{emma}}~\ref{Lemma 3} for all $\matr{x}$ with the same weight, ${\rm{P}}(\matr{x}^T\matr{R}\matr{x}<1)$ are equal. Thus, \eqref{boundevenbinary} can be written as follows
\begin{equation}\label{9}
{\rm{P}}(E_K)\leq \sum_{j=1}^{\lfloor\frac{K}{2}\rfloor}{{K\choose 2j} 2^{2j} {\rm{P}}(\matr{x}_j^T\matr{R}\matr{x}_j<1)},
\end{equation}
where $\matr{x}_j$ is an arbitrary vector with weight $2j$. 
Using {L\textsc{emma}}~\ref{Lemma 2} in \eqref{9} results in
\begin{equation}\label{10}
{\rm{P}}(E_K)\leq  \sum_{j=1}^{\lfloor\frac{K}{2}\rfloor}{{K\choose 2j} 2^{2j} {\rm{P}}(B_j)}=\sum_{j=1}^{\lfloor\frac{K}{2}\rfloor}{{K\choose 2j} 2^{2j} {\sum_{i=0}^{\frac{N}{4}-1} {{N\choose i}{\rm{p}}(j)^{N-i}(1-{\rm{p}}(j))^i} }},
\end{equation}
where
\begin{equation}\label{11}
{\rm{p}}(j)={\rm{P}}(u_\ell(\matr{x}_j)=0)={2j\choose j}2^{-2j}.
\end{equation}
The Binomial distribution function $ f(i)={{N\choose i}{\rm{p}}(j)^{N-i}(1-{\rm{p}}(j))^i} $ is an increasing function for $i<i_{\rm m}\eqdef \lfloor N(1-{\rm{p}}(j))\rfloor$. Furthermore, in Appendix \ref{B} it is proven that $ \left\lfloor N\left( 1-{\rm{p}}(j)\right) \right\rfloor \geq  \left \lfloor\frac{N}{2}\right \rfloor$. Therefore, an upper bound for \eqref{10} is derived as 
\begin{eqnarray}\label{13}
{\rm{P}}(E_K)&\leq&\sum_{j=1}^{\lfloor\frac{K}{2}\rfloor}{{K\choose 2j} 2^{2j} {\sum_{i=0}^{\frac{N}{4}-1} {{N\choose i}{\rm{p}}(j)^{N-i}(1-{\rm{p}}(j))^i} }}\nonumber \\
&\leq&
\sum_{j=1}^{\lfloor\frac{K}{2}\rfloor}{{K\choose 2j} 2^{2j} \frac{N}{4}{{{N\choose \frac{N}{4}}{\rm{p}}(j)^{N-\frac{N}{4}}(1-{\rm{p}}(j))^{\frac{N}{4}}} }}.
\end{eqnarray}
To simplify more, the following inequality is used 
\begin{eqnarray}\label{14A}
{m\choose r}\leq 2^{m{\rm{h}}\left(\frac{r}{m}\right)},
\end{eqnarray}
where
\begin{eqnarray}\label{15}
{\rm{h}}(t)=-t\log_2t-(1-t)\log_2(1-t),
\end{eqnarray}
denotes the binary entropy function.
The proof of \eqref{14A} is given in Appendix \ref{C}. By using \eqref{14A}, \eqref{13} can be written as
 \begin{equation}\label{15A}
{\rm{P}}(E_K)\leq
\sum_{j=1}^{\lfloor\frac{K}{2}\rfloor}{  \frac{N}{4} 2^{K\left({\rm{h}}\left(\frac{2j}{K}\right)+\frac{2j}{K}\right)} \left[2^{4{\rm{h}}\left( \frac{1}{4}\right)}{\rm{p}}(j)^{3}(1-{\rm{p}}(j))\right]^{N/4} }.
\end{equation}

We wish to prove that the bound in \eqref{15A} converges to 0 in the limit $K,N\rightarrow \infty$ while $\zeta=\frac{K}{N\log_3K}$ is kept less than $\frac{3}{8}$. In this regard, we divide the range of summation $j=1,\cdots,{\lfloor\frac{K}{2}\rfloor}$ into two, as follows
 \begin{eqnarray}\label{15A1}
{\rm{P}}(E_K)&\leq& S_1+S_2, \\
S_1&=&\sum_{j=1}^{j_0}{  \frac{N}{4} 2^{K\left({\rm{h}}\left(\frac{2j}{K}\right)+\frac{2j}{K}\right)} \left[2^{4{\rm{h}}\left( \frac{1}{4}\right)}{\rm{p}}(j)^{3}(1-{\rm{p}}(j))\right]^{N/4} }, \\
S_2&=&\sum_{j=j_0+1}^{\lfloor\frac{K}{2}\rfloor}{  \frac{N}{4} 2^{K\left({\rm{h}}\left(\frac{2j}{K}\right)+\frac{2j}{K}\right)} \left[2^{4{\rm{h}}\left( \frac{1}{4}\right)}{\rm{p}}(j)^{3}(1-{\rm{p}}(j))\right]^{N/4} }, \label{15A12}
\end{eqnarray}
where $j_0=\lfloor\frac{K}{2(\log_2K)^u}\rfloor$ and $u>1$ is a constant.

We first show that the sum $S_1$ tends to 0. We use the bound 
 \begin{eqnarray}\label{15A2}
{\rm h}(t)+t\leq -t\left(\log_2t -\log_2(2{\rm e}) \right),
\end{eqnarray}
which is tight as $t\rightarrow +0$. Since the right-hand side of \eqref{15A2} is an increasing function of $t$, one has, for $j=1,\cdots,j_0$,
 \begin{eqnarray}\label{15A3}
h\left(\frac{2j}{K}  \right)+\frac{2j}{K}\leq \frac{u \log_2\log_2K}{(\log_2K)^u}+\frac{\log_2(2{\rm e})}{(\log_2K)^u}.
\end{eqnarray}
On the other hand, $p(j)$ is a decreasing function of $j$ with $p(1)=\frac{1}{2}$. Furthermore, for $0\leq x \leq \frac{1}{2}$, the function $x^3(1-x)$ is an increasing function of $x$. Thus the factor $p(j)^3\left(1-p(j) \right)$ can be bounded from above by $\frac{1}{2^4}$, so that one has 
\begin{eqnarray}
2^{4{\rm{h}}\left( \frac{1}{4}\right)}{\rm{p}}(j)^{3}(1-{\rm{p}}(j))\leq 2^{4\left({\rm h}\left(\frac{1}{4} \right)-1 \right)}.
\end{eqnarray}
One thus has 
\begin{eqnarray} \label{S11}
S_1\leq  \frac{KN}{8\left ( \log_2K \right)^u} 2^{\frac{u K\log_2\log_2K}{(\log_2K)^u}+\frac{K\log_2(2{\rm e})}{(\log_2K)^u}+N\left( {\rm h}\left(\frac{1}{4} \right)-1  \right)}
\end{eqnarray}
As $K\rightarrow \infty$, the following holds
\begin{eqnarray} \label{trik}
\frac{\log_2\log_2K}{\left( \log_2K\right)^u}=o\left( \frac{1}{\log_2K}\right). 
\end{eqnarray}
To show that it holds, consider the behavior of 
\begin{eqnarray}
\frac{\log_2\log_2K}{\left( \log_2K\right)^{u-1}}
\end{eqnarray}
as $K\rightarrow \infty$. Let $\kappa\eqdef \left( \log_2K\right)^{u-1} $. Then one has $\kappa\rightarrow \infty$ as $K\rightarrow \infty$, and 
\begin{eqnarray}
\frac{\log_2\log_2K}{\left( \log_2K\right)^{u-1}}=\frac{\log_2\kappa}{(u-1)\kappa}\rightarrow 0,
\end{eqnarray}
which proves \eqref{trik}. Using \eqref{trik} results in that the dominant term in the exponent of \eqref{S11} is $N\left( {\rm h}\left(\frac{1}{4} \right)-1  \right)$, which tends to $-\infty$ as $K\rightarrow \infty$. This proves that $S_1\rightarrow0$ holds. 

We next show that $S_2$ also tends to 0 provided that $\zeta<\frac{3}{8}$ holds. To this end we use some bounds as follows. Since the function ${\rm h}(t)+t$ takes its maximum at $t=\frac{2}{3}$, one has 
 \begin{eqnarray}
h\left(\frac{2j}{K}  \right)+\frac{2j}{K}\leq  h\left(\frac{2}{3}  \right)+\frac{2}{3}=\log_23.
\end{eqnarray} 
Next we use the following bound 
 \begin{eqnarray}\label{22}
{\rm{p}}\left(j\right)={2j\choose{j}} 2^{-2j}\leq \frac{{\rm e}}{\pi \sqrt{2j}},
\end{eqnarray}
which is proven in Appendix~\ref{D}, as well as the fact that $p(j)$ is a decreasing function of $j$. Thus one has, for  $j=j_0+1,\cdots,\lfloor\frac{K}{2}\rfloor$,
 \begin{eqnarray} \label{bound3be2}
p(j)^3(1-p(j))\leq \left(\frac{\pi^2K}{{\rm e}^2\left(\log_2K \right)^u}   \right)^{-3/2}.
\end{eqnarray}
Applying the presented bounds, one can write \eqref{15A12} as 
 \begin{eqnarray}\label{upperboundI2}
S_2\leq \left(\frac{K}{2}-\frac{K}{2\left(\log_2K\right)^u}+1\right) \frac{N}{4} 2^{K\log_23-\left(3N/8\right)\left (\log_2K+2\log_2\left(\pi/{\rm e}\right) -u\log_2\log_2K  \right)}.
\end{eqnarray}
The right-hand side vanishes in the limit $K\rightarrow \infty$ provided that $N$ goes to infinity in such a way as to satisfy 
\begin{eqnarray}
\zeta=\frac{K \log_23}{N\log_2K}<\frac{3}{8}.
\end{eqnarray}
This proves that $S_2\rightarrow 0$ holds when $\zeta <3/8$. 
Therefore, we have
\begin{eqnarray}\label{31}
\lim_{K\rightarrow\infty}{\rm{P}}(E_K)=0,
\end{eqnarray}
if
\begin{equation} 
 \zeta<\frac{3}{8}.
 \end{equation}
 Finally, \eqref{S11}, \eqref{upperboundI2} yield 
\begin{equation}
\sum_{K=1}^{\infty} {\rm{P}}(E_K)<\infty,
\end{equation}
if $\zeta<\frac{3}{8}$. Therefore, by using the Borel-Cantelli lemma \cite{bremaud1988introduction} and the fact that the maximum possible value for $\eta$ is $1$, it is concluded that $\eta$ converges to $1$ almost surely if $\zeta<\frac{3}{8}$.
\end{proof}

In the next theorem we derive a lower bound for the optimum asymptotic multiuser efficiency of binary antipodal randomly spread CDMA for $\frac{3}{8}\leq\zeta<\frac{1}{2}$. \\
\begin{theorem}\label{theorem2}
 The optimum asymptotic multiuser efficiency is greater than $\gamma\in (0,1)$ almost surely as $K,N\rightarrow \infty$, if
$\zeta=\frac{K}{N\log_3{K}}$ is kept less than $(4-\gamma)/8$.
\end{theorem}
\begin{proof}
To prove T\textsc{heorem}~\ref{theorem2}, we first introduce a generalized form of L\textsc{emma}~\ref{Lemma 2} as L\textsc{emma}~\ref{Lemma 4}.\\
\begin{lemma} \label{Lemma 4}
 Let $\matr{x}_j$ be a vector with weight $2j$ and $V_j$ be the event that the number of nonzero elements of $\matr{u}(\matr{x}_j)=[u_1(\matr{x}_j),u_2(\matr{x}_j),\cdots,u_N(\matr{x}_j)]^T$ is less than $\frac{N\gamma}{4}$ where $\gamma\in (0,1)$. Then, ${\rm{P}}(\matr{x}_j^TR\matr{x}_j<\gamma)\leq {\rm{P}}(V_j)$.  
\end{lemma}

Let $E_{K,\gamma}$ be the event that $\matr{x}^T\matr{R}\matr{x}<\gamma$ for at least one nonzero error vector $\matr{x}$. Then, a similar procedure as in the proof of T\textsc{heorem}~\ref{theorem1} is used. Moreover, in the proof of T\textsc{heorem}~\ref{theorem2},  L\textsc{emma}~\ref{Lemma 4} is used instead of L\textsc{emma}~\ref{Lemma 2}. 
One has
\begin{eqnarray}
{\rm{P}}(E_{K,\gamma})&=&{\rm{P}}\left(\bigcup_{\matr{x}\in \{\pm1,0\}^K\setminus\{\matr{{0}}\}} \matr{x}^T\matr{R}\matr{x}<\gamma\right)\nonumber \\
&\leq&\sum_{\matr{x}\in \{\pm1,0\}^K\setminus\{\matr{{0}}\}}{{\rm{P}}(\matr{x}^T\matr{R}\matr{x}<\gamma)} \nonumber \\
&=&\sum_{j=1}^{\lfloor\frac{K}{2}\rfloor}{{K\choose 2j} 2^{2j} {\rm{P}}(\matr{x}_j^T\matr{R}\matr{x}_j<\gamma)} \nonumber \\
&\leq& \sum_{j=1}^{\lfloor\frac{K}{2}\rfloor}{{K\choose 2j} 2^{2j} {\rm{P}}(V_j)} \nonumber \\
&\leq&\sum_{j=1}^{\lfloor\frac{K}{2}\rfloor}{{K\choose 2j} 2^{2j} {\sum_{i=0}^{\frac{N\gamma}{4}-1} {{N\choose i}{\rm{p}}(j)^{N-i}(1-{\rm{p}}(j))^i} }}\nonumber \\
&\leq&\sum_{j=1}^{\lfloor\frac{K}{2}\rfloor}{{K\choose 2j} 2^{2j} \frac{N\gamma}{4}{{{N\choose \frac{N\gamma}{4}}{\rm{p}}(j)^{N-\frac{N\gamma}{4}}(1-{\rm{p}}(j))^{\frac{N\gamma}{4}}} }} \nonumber \\
&\leq&
\sum_{j=1}^{\lfloor\frac{K}{2}\rfloor}{  \frac{N\gamma}{4} 2^{K\left({\rm{h}}\left(\frac{2j}{K}\right)+\frac{2j}{K}\right)} \left[2^{4{\rm{h}}\left( \frac{\gamma}{4}\right)}{\rm{p}}(j)^{4-\gamma}(1-{\rm{p}}(j))^\gamma\right]^{N/4} }.
\end{eqnarray}
We bound the right-hand side of the above inequality in the same way as in the proof of T\textsc{heorem}~\ref{theorem1}. Specially, we have 
\begin{eqnarray}\label{th2S34}
{\rm{P}}(E_{K,\gamma})&\leq& S_{1,\gamma}+S_{2,\gamma}, \\
S_{1,\gamma}&=&\sum_{j=1}^{j_0}{  \frac{N\gamma}{4} 2^{K\left({\rm{h}}\left(\frac{2j}{K}\right)+\frac{2j}{K}\right)} \left[2^{4{\rm{h}}\left( \frac{\gamma}{4}\right)}{\rm{p}}(j)^{4-\gamma}(1-{\rm{p}}(j))^\gamma\right]^{N/4} }, \\
S_{2,\gamma}&=&\sum_{j=j_0+1}^{\lfloor\frac{K}{2}\rfloor}{  \frac{N\gamma}{4} 2^{K\left({\rm{h}}\left(\frac{2j}{K}\right)+\frac{2j}{K}\right)} \left[2^{4{\rm{h}}\left( \frac{\gamma}{4}\right)}{\rm{p}}(j)^{4-\gamma}(1-{\rm{p}}(j))^\gamma\right]^{N/4} },
\end{eqnarray}
where $j_0=\lfloor\frac{K}{2(\log_2K)^u}\rfloor$ and $u>1$ is a constant.
We will show that both $S_{1,\gamma}$ and $S_{2,\gamma}$ tend to 0 as $K,N\rightarrow \infty$ while $\zeta$ is kept less than $\frac{4-\gamma}{8}$.

We first show that $S_{1,\gamma}\rightarrow 0$ holds. since one has 
\begin{eqnarray}
\frac{d}{dx}x^{4-\gamma}(1-x)^\gamma&=&(4-\gamma)x^{3-\gamma}(1-x)^\gamma-\gamma x^{4-\gamma}(1-x)^{\gamma-1} \nonumber \\
&=&x^{3-\gamma}(1-x)^{\gamma-1}\left((4-\gamma)(1-x)-\gamma x \right) \nonumber \\
&=&x^{3-\gamma} (1-x)^{\gamma-1} \left( (4-\gamma)-4x \right),
\end{eqnarray}
the function $x^{4-\gamma}(1-x)^\gamma$ is an increasing function for $x\in \left[0,1-\frac{\gamma}{4}\right. \left. \right)$.
Note that the interval $\left[0,1-\frac{\gamma}{4}\right. \left. \right)$ contains 1/2 whenever $\gamma\in(0,1)$. Thus, the factor $p(j)^{4-\gamma}(1-p(j))^\gamma$ can be bounded from above by $1/2^4$. The proof that $S_{1,\gamma}\rightarrow 0$ holds is complete by observing 
\begin{eqnarray}
2^{4{\rm h}\left( \frac{\gamma}{4} \right)} p(j)^{4-\gamma}(1-p(j))^\gamma \leq 2^{4\left({\rm h}\left( \frac{\gamma}{4}\right)-1  \right)}.
\end{eqnarray}
In order to prove that $S_{2,\gamma}\rightarrow 0$ holds, we use instead of \eqref{bound3be2} the following inequality
 \begin{eqnarray} \label{bound3be2new}
p(j)^{4-\gamma}(1-p(j))^\gamma \leq \left(\frac{\pi^2K}{{\rm e}^2\left(\log_2K \right)^u}   \right)^{-(4-\gamma)/2}.
\end{eqnarray}
Repeating the same argument as in the proof of T\textsc{heorem}~\ref{theorem1}, one can show that $S_{2,\gamma}\rightarrow 0$ holds when $\zeta<(4-\gamma)/8$.
Therefore, it is obtained that $\eta$ is greater than $\gamma$ almost surely if
\begin{equation}\label{33}
\zeta<(4-\gamma)/8
\end{equation}
which proves T\textsc{heorem}~\ref{theorem2}.
\end{proof}

From T\textsc{heorem}~\ref{theorem2}, a lower bound for $\eta$ is obtained as
\begin{equation}\label{33A}
\eta>4(1-2\zeta).
\end{equation}
Fig. 1 shows the obtained results in T\textsc{heorem}~\ref{theorem1} and T\textsc{heorem}~\ref{theorem2}. For $\zeta\leq\frac{3}{8}$ the result is exact and for
$\frac{3}{8}<\zeta<\frac{1}{2}$, the curve is a lower bound for the optimum asymptotic multiuser efficiency.

\begin{figure}
\centering
\includegraphics[width=5in]{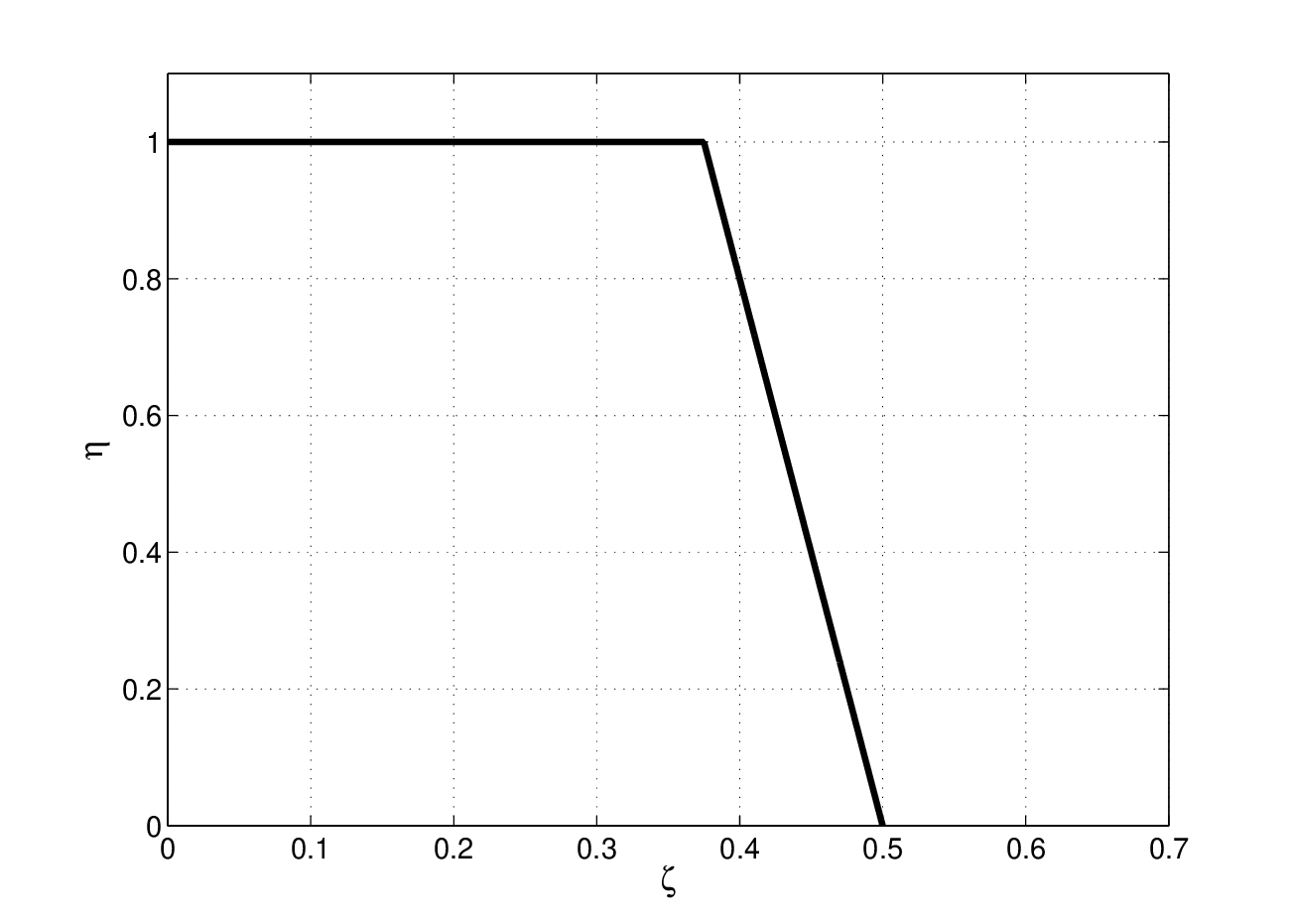}
\caption{The optimum asymptotic multiuser efficiency lower bound versus $\zeta=\frac{K}{N\log_3K}$.}
\end{figure}

\subsection{The optimum asymptotic multiuser efficiency for i.i.d. Gaussian spreading}
In this section we investigate the optimum asymptotic multiuser efficiency for a randomly spread CDMA when the entries of $H$ are i.i.d. Gaussian distributed.

\begin{theorem}\label{theorem3}
 Let $H_{i,j}\sim \mathcal{N}(0,\frac{1}{{N}})$.
The optimum asymptotic multiuser efficiency converges to $1$ almost surely as $K,N\rightarrow \infty$, if $\frac{K}{N\log_3{K}}$ is kept less than $\frac{1}{2}$.
\end{theorem}
\begin{proof}
 By using \eqref{uniobound} and L\textsc{emma}~\ref{Lemma 4}
\begin{equation}\label{unionG}
{\rm{P}}(E_K)\leq \sum_{j=1}^{K}{{K\choose j} 2^{j} {\rm{P}}(\matr{x}_j^T\matr{R}\matr{x}_j< 1)},
\end{equation}
where $\matr{x}_j$ is an arbitrary vector with weight $j$. For sake of simplicity, we write \eqref{unionG} as
\begin{eqnarray}\label{Gdefine}
{\rm{P}}(E_K)
\leq  \underbrace{2K {\rm{P}}(\matr{x}_1^T\matr{R}\matr{x}_1< 1)}_{\eqdef G_1}+
\underbrace{\sum_{j=2}^{K}{{K\choose j} 2^{j} {\rm{P}}(\matr{x}_j^T\matr{R}\matr{x}_j< 1)}}_{\eqdef G_2},
\end{eqnarray}
where $\matr{x}_1$ is an arbitrary vector with weight 1. From \cite[eq. (21)]{tse2000optimum}, the term ${\rm{P}}(\matr{x}_1^T\matr{R}\matr{x}_1< 1)$ decays exponentially in $N$. Since we assume $\zeta=\frac{K}{N\log_3 K}$ is fixed, it can be written that
\begin{eqnarray}\label{G1calc}
G_1={\rm O}\left(K{\rm e}^{-\alpha \frac{ K}{\log_3 K}}\right),
\end{eqnarray}
where $\alpha$ is a finite positive real number. 

The term $G_2$ is calculated as follows.
It can be shown that conditioned on weight $j$
\begin{equation}\label{35}
u_\ell(\matr{x}_j)\sim \mathcal{N}\left(0,\frac{j}{N}\right).
\end{equation}
 Therefore, from \eqref{1}, $\frac{N}{j}(\matr{x}_j^T\matr{R}\matr{x}_j)$ has a chi-squared distribution with $N$ degrees of freedom. Therefore, 
\begin{equation}\label{36}
{\rm{P}}(\matr{x}_j^T\matr{R}\matr{x}_j< 1)=\int_{0}^{\frac{N}{j}} {\frac{1}{2^{N/2}\Gamma(N/2)}x^{\frac{N}{2}-1}\exp{(-x/2)}}dx.
\end{equation}
A chi-squared distribution with $N$ degrees of freedom is an increasing function in $[0,N-2]$ for $N>2$. Therefore, since $j\geq 2$, the term inside of the integration in \eqref{36} 
is an increasing function. Thus,
 \begin{equation}\label{37}
{\rm{P}}(\matr{x}_j^T\matr{R}\matr{x}_j<1)\leq \frac{N}{j 2^{N/2}\Gamma(N/2)}\left(\frac{N}{j}\right)^{\frac{N}{2}-1}\exp{\left(-\frac{N}{2j}\right)}.
\end{equation}
Without loss of generality we assume that $N$ is an even integer. Based on Stirling's formula a lower bound for $\Gamma(N/2)$ is
 \begin{eqnarray}\label{38}
\Gamma{\left(N/2\right)}=\left(N/2-1\right)!=\frac{\left(N/2\right)!}{N/2}> 2\sqrt{\pi/ N}\left(\frac{N}{2e}\right)^{N/2},
\end{eqnarray}
Therefore, 
 \begin{equation}\label{39}
G_2\leq \sum_{j=2}^{K} \frac{1}{2} \sqrt{\frac{N}{\pi}}{K\choose j} 2^{j}
\left( \frac{{\rm e}^{1-\frac{1}{j}}}{j}\right)^{\frac{N}{2}}.
\end{equation}
Applying the bounds \eqref{14A}, we divide the summation as follows
 \begin{equation}
G_2\leq \sum_{j=2}^{K}\frac{1}{2} \sqrt{\frac{N}{\pi}} 2^{K\left({\rm{h}}\left(\frac{j}{K}\right)+\frac{j}{K}\right)} \left( \frac{{\rm e}^{1-\frac{1}{j}}}{j}\right)^{\frac{N}{2}}=S_3+S_4,
\end{equation}
where 
 \begin{eqnarray}
S_3&=& \sum_{j=2}^{j_1}\frac{1}{2} \sqrt{\frac{N}{\pi}} 2^{K\left({\rm{h}}\left(\frac{j}{K}\right)+\frac{j}{K}\right)} \left( \frac{{\rm e}^{1-\frac{1}{j}}}{j}\right)^{\frac{N}{2}}, \\
S_4&=& \sum_{j=j_1+1}^{K}\frac{1}{2} \sqrt{\frac{N}{\pi}} 2^{K\left({\rm{h}}\left(\frac{j}{K}\right)+\frac{j}{K}\right)} \left(\frac{{\rm e}^{1-\frac{1}{j}}}{j} \right)^{\frac{N}{2}}, \\
j_1&=&\left\lfloor\frac{K}{(\log_2K)^u}\right\rfloor,
\end{eqnarray}
and $u>1$ is a constant. Since $\frac{{\rm e}^{1-\frac{1}{j}}}{j}$ is a decreasing function of $j$, one has $\frac{{\rm e}^{1-\frac{1}{j}}}{j} \leq \frac{\sqrt{\rm e}}{2}$ for $j\geq 2$.
Using the same argument as that in the previous proofs, we can prove $S_3\rightarrow 0$ by bounding $S_3$ as
\begin{eqnarray} 
S_3\leq \frac{K}{2(\log_2K)^u} \sqrt{\frac{N}{\pi}} 2^{\frac{u K\log_2\log_2K}{(\log_2K)^u}+\frac{K\log_2(2{\rm e})}{(\log_2K)^u}+(N/4)\log_2\left({\rm e}/4\right)} \rightarrow 0.
\end{eqnarray}

One can also bound $S_4$ as 
\begin{eqnarray} 
S_4&\leq& \left(K-\frac{K}{(\log_2K)^u} +1 \right)\frac{1}{2} \sqrt{\frac{N}{\pi}} 2^{K\log_23}\left(\frac{{\rm e} (\log_2K)^u}{K} \right)^{N/2} \nonumber \\
&=& \left(K-\frac{K}{(\log_2K)^u} +1 \right)\frac{1}{2} \sqrt{\frac{N}{\pi}} 2^{K\log_23 -(N/2)\log_2K +(N/2)u\log_2\log_2K+(N/2)\log_2{\rm e}}.
\end{eqnarray}
The above upper bound tends to 0 as $K\rightarrow \infty$ if $N$ goes to infinity in such a way as to satisfy 
\begin{eqnarray} 
\zeta=\frac{K}{N\log_3K}<\frac{1}{2}.
\end{eqnarray}
Furthermore, it can be shown that
\begin{eqnarray}\label{BorelGaussian}
\sum_{K=1}^{+\infty}{\rm{P}}(E_K)<\infty,
\end{eqnarray}
which together with the application of the Borel-Cantelli lemma proves that $\eta$ converges to 1 almost surely if $K \rightarrow \infty$ and $\zeta$ is kept less than $\frac{1}{2}$.
\end{proof}

Tse and Verd\'{u} in \cite{ tse2000optimum} prove that the optimum asymptotic multiuser efficiency of a CDMA system with general i.i.d. entries converges to 1 when $K\rightarrow\infty$ and $\frac{K}{N}$ is kept finite. However, in T\textsc{heorem} 1, T\textsc{heorem} 2 and T\textsc{heorem} 3 we prove that the loading factor, $\frac{K}{N}$, can grow logarithmically with $K$ for binary antipodal and Gaussian matrices. In fact, the result in \cite{tse2000optimum} can be obtained for random binary antipodal and random Gaussian spreading when $\zeta\rightarrow 0$ as presented here.

\section{Random Detecting Matrices}
In this section we study detecting matrices. They are closely connected to optimum asymptotic multiuser efficiency. Detecting matrices originate from the coin weighing problem in mathematics \cite{soderberg1963combinatory}-\cite{mow2009recursive}.
Let $\mathcal{S}$ be a subset of $\mathbb{R}$. For a given data set $\mathcal{S}$ such that $\matr{x}_1,\matr{x}_2\in \mathcal{S}^K$, an $N \times K$ matrix $\matr{H}$ is called detecting if and only if
\begin{equation}\label{43}
\matr{H}\matr{x}_1=\matr{H}\matr{x}_2 \Rightarrow \matr{x}_1=\matr{x}_2,
\end{equation}
where $\matr{x}_1$ and $\matr{x}_2$ are $K \times 1$ vectors. Another representation form of \eqref{43} is
\begin{equation}\label{44}
\matr{H}\matr{x}=\matr{0}_{\{N\times 1\}} \Rightarrow \matr{x}=\matr{0}_{\{K\times 1\}},
\end{equation}
where $\matr{x}\in\left\{\mathcal{S}^K-\mathcal{S}^K\right\}$ in which
\begin{equation}
\left\{\mathcal{S}^K-\mathcal{S}^K\right\}=\left\{\matr{x}_1-\matr{x}_2|\matr{x}_1,\matr{x}_2\in\mathcal{S}^K\right\}.
\end{equation}
 One can write \eqref{44} as
\begin{equation}\label{45}
{\rm{Null}}(\matr{H})\bigcap \left\{\mathcal{S}^K-\mathcal{S}^K\right\}=\{\matr{0}_{\{K\times 1\}}\},
\end{equation}
where ${\rm{Null}}(\matr{H})$ is the null space of $H$. 
For any binary input $\mathcal{S}=\{a,b\}$, $a\neq b$, a matrix $\matr{H}$ is detecting if and only if
 \begin{equation}\label{detectingbinarykol}
{\rm{Null}}(\matr{H})\bigcap \{\pm (a-b),0\}^K=\{\matr{0}_{\{K\times 1\}}\}.
\end{equation}
Note that this is equivalent to
\begin{equation}\label{detectingbinary}
{\rm{Null}}(\matr{H})\bigcap \{\pm 1,0\}^K=\{\matr{0}_{\{K\times 1\}}\}.
\end{equation}

From \eqref{5} and \eqref{detectingbinary}, it can be observed that there is a connection between $\eta$ and the concept of detecting matrices. In fact if in a CDMA system the spreading matrix, $\matr{H}$, is not detecting then there is an error vector $\matr{x}\neq \matr{0}_{\{K \times 1\}}$ such that $\matr{x}^T\matr{R}\matr{x}=0$. Therefore, if the spreading matrix is not a detecting matrix then the optimum asymptotic multiuser efficiency is equal to $0$.

In \cite{lindstrom1966combinatorial}, it is proven that
\begin{equation}\label{46}
\lim _{K\rightarrow \infty}{\frac{N_0 \log_{2}{K}}{K}}=2,
\end{equation}
where $N_0$ is the minimum possible of $N$ such that an $N\times K$  binary $\{0,1\}$ or binary antipodal $\{\pm 1\}$ detecting matrix exists for any binary input \cite{du1993combinatorial}. Therefore, it is concluded that the optimum asymptotic multiuser efficiency is equal to 0 when $K\rightarrow \infty$ and $\zeta=\frac{K}{N\log_3K}$ is kept greater than $\frac{\log_23}{2}$. This result can be considered as a supplementary result to T\textsc{heorem}~\ref{theorem1} and T\textsc{heorem}~\ref{theorem2}. 
Note that there is no result for the optimum asymptotic multiuser efficiency of a random binary antipodal spread CDMA in $\zeta\in \left( \frac{1}{2},\frac{\log_23}{2}\right)$ so far. 

From \eqref{5} and \eqref{detectingbinary} it is also observed that if the optimum asymptotic multiuser efficiency is greater than 0 then the spreading matrix is a detecting matrix. Therefore, from Fig. 1 it can be concluded that a binary antipodal random matrix is detecting if $K\rightarrow \infty$ and $\zeta=\frac{K}{N\log_3K}$ is kept less than $\frac{1}{2}$. This was also proven by  Erd\H{o}s and R\'{e}nyi in 1963 \cite{erdgs1963two}. However, the converse statement might not be true. This means that a spreading matrix can be detecting but $\eta$ might be vanishing in the large system limit. Therefore, being detecting cannot be considered isomorphic to non-vanishing $\eta$. However, it is beneficial to know under which condition a spreading matrix is detecting.
In the next theorem we generalize the result by Erd\H{o}s and R\'{e}nyi for an i.i.d. random matrix whose elements are chosen randomly from a finite set.

\begin{theorem}\label{theorem4}
 Let $\Sigma=\{ d_1,d_2,\cdots d_{\rm m}\}$ be a symmetric finite set and $\psi$ a zero mean symmetric non-degenerate probability
distribution on $\Sigma$.
An $N\times K$ matrix $\matr{H}$ whose elements are i.i.d. with distribution $\psi$ is a detecting matrix for any binary input set $\mathcal{S}\in\{a,b\}$, $a\neq b$, if
$K,N\rightarrow \infty$ and $\zeta=\frac{K}{N\log_3K}$ is kept less than $\frac{{\rm{rank}}(\Sigma)}{2}$,
where ${\rm{rank}}(\Sigma)$ denotes the dimension of $\Sigma$ as a set of vectors over the field of rational numbers $\mathbb{Q}$,
i.e., the maximum size of a subset of $\Sigma$ such that no nontrivial rational linear
combination of which vanishes.
\end{theorem}
\begin{proof}
Suppose that $D_K$ is the event that $\matr{H}$ is a detecting matrix.
From \eqref{detectingbinary}, $\matr{H}$ is not a detecting matrix if there is at least one  $\matr{x}\in \{\pm 1,0\}^K\setminus \left\{{\matr{0}_{\{K\times 1\}}}\right\}$ such that $\matr{H}\matr{x}=\matr{0}_{\{N\times 1\}}$.
By using the union bound, a lower bound for ${\rm{P}}(D_K)$ is obtained as
\begin{equation}\label{47}
{\rm{P}}(D_K)=1-{\rm{P}}(\bar{D}_K)\geq 1-\sum_{\matr{x}\in \{\pm 1,0\}^K\setminus \left\{{\matr{0}}\right\}  }{{\rm{P}}(\matr{H}\matr{x}=\matr{0}_{\{N\times 1\}})},
\end{equation}
where $\bar{D}_K$ is the complement of $D_K$.
From L\textsc{emma}~\ref{Lemma 4}, ${\rm{P}}(\matr{H}\matr{x}=\matr{0}_{\{N\times 1\}})$ is equal for all $\matr{x}\in \{\pm 1,0\}^K$ with the same weight. Therefore, \eqref{47} can be written as
\begin{equation}\label{48}
{\rm{P}}(D_K)\geq 1-\sum_{j=1}^{K}{{K \choose j} 2^{j}{\rm{P}}(\matr{H}\matr{x}_j=\matr{0}_{\{N\times 1\}})}.
\end{equation}
where $\matr{x}_j$ is a vector with weight $j$.
 Entries of $\matr{H}$ are i.i.d.. Therefore, ${\rm{P}}(\matr{H}\matr{x}_j=\matr{0}_{\{N\times 1\}})$ in \eqref{48} can be written as
\begin{eqnarray}\label{49}
{\rm{P}}(\matr{H}\matr{x}_j=\matr{0}_{\{N\times 1\}})=\left({\rm{P}}\left(\sum_{i=1}^{K}{H_{ri}x_j(i)}=0\right)\right)^N,
\end{eqnarray}
where $r\in\{1,2,\cdots,N\}$. Hence,
\begin{equation}\label{50}
{\rm{P}}(D_K)\geq 1-\sum_{j=1}^{K}{{K \choose j} 2^{j}\left({\rm{P}}\left(\sum_{i=1}^{K}{H_{ri}x_j(i)}=0\right)\right)^N}.
\end{equation}
Eq. \eqref{50} can be simplified more as
\begin{equation}\label{joda1}
{\rm{P}}(D_K)\geq 1-2K\left({\rm{P}}(H_{r1}=0)\right)^K-\sum_{j=2}^{K}{{K \choose j} 2^{j}\left({\rm{P}}\left(\sum_{i=1}^{K}{H_{ri}x_j(i)}=0\right)\right)^N}.
\end{equation}
Then, based on \cite[Lemma 3]{alishahi2012design}
\begin{equation}\label{51}
{\rm{P}}\left(\sum_{i=1}^{K}{H_{ri}x_j(i)}=0\right)={\rm{O}}\left(j^{-\frac{{\rm{rank}}(\Sigma)}{2}}\right).
\end{equation}
Therefore,
\begin{equation}\label{52}
{\rm{P}}(D_K)\geq 1-2K\left({\rm{P}}(H_{r1}=0)\right)^K-\sum_{j=2}^{K}{{K \choose j} 2^{j}{\rm{O}}\left(j^{-\frac{N{{\rm{rank}}}(\Sigma)}{2}}\right)}.
\end{equation}
We divide the range of summation as 
\begin{equation}\label{splitth4}
{\rm{P}}(D_K)\geq 1-S_5-S_6-S_7,
\end{equation}
where
\begin{eqnarray}
S_5&=&2K\left({\rm{P}}(H_{r1}=0)\right)^K,\\
S_6&=&\sum_{j=2}^{j_1}{{K \choose j} 2^{j}{\rm{O}}\left(j^{-\frac{N{{\rm{rank}}}(\Sigma)}{2}}\right)}, \\
S_7&=&\sum_{j=j_1+1}^{K}{{K \choose j} 2^{j}{\rm{O}}\left(j^{-\frac{N{\rm{rank}}(\Sigma)}{2}}\right)},\\
j_1&=&\left\lfloor\frac{K}{(\log_2K)^u}\right\rfloor,
\end{eqnarray}
 and $u>1$ is a constant. 

 Since the pdf of the elements of $\matr{H}$ is a non-degenerate pdf, ${\rm{P}}(H_{r1}=0)$ is less than $1$. Thus, $S_5$ tends to zero when $K\rightarrow +\infty$. Next we will show that $S_6$ and $S_7$ tend to 0 as $K,N\rightarrow \infty$ while $\zeta$ is kept less than $\frac{{\rm{rank}}(\Sigma)}{2}$. 

We first show that $S_6\rightarrow 0$ holds. Using \eqref{15A2}, one can bound $S_6$ as 
\begin{eqnarray}
S_6&\leq& \sum_{j=2}^{j_1}{ 2^{K\left({\rm{h}}\left(\frac{j}{K}\right)+\frac{j}{K}\right)}   {\rm{O}}\left(j^{-\frac{N{{\rm{rank}}}(\Sigma)}{2}}\right)} \nonumber \\
&\leq& \frac{cK}{\left ( \log_2K \right)^u} 2^{\frac{u K\log_2\log_2K}{(\log_2K)^u}+\frac{K\log_2(2{\rm e})}{(\log_2K)^u}-\frac{N{\rm rank}(\Sigma)}{2} },
\end{eqnarray}
where $c$ is a finite constant. The dominant term in the exponent is $-\frac{N{\rm rank}(\Sigma)}{2} $, which tends to $-\infty$ as $K,N\rightarrow \infty$. This proves that $S_6\rightarrow 0$ holds.

We next show that $S_7$ also tends to 0 provided that $\zeta <\frac{{\rm{rank}}(\Sigma)}{2}$. Since the function ${\rm h}(t)+t$ takes its maximum at $t=2/3$, one has
\begin{eqnarray}
S_7&\leq& \sum_{j=j_1+1}^{K}{ 2^{K\left({\rm{h}}\left(\frac{j}{K}\right)+\frac{j}{K}\right)}   {\rm{O}}\left(j^{-\frac{N{{\rm{rank}}}(\Sigma)}{2}}\right)} \nonumber \\
&\leq& \acute{c}\left( K-\frac{K}{\left ( \log_2K \right)^u} -1 \right)2^{K\log_23 -\frac{N{\rm rank}(\Sigma)}{2} \left(\log_2K-u\log_2\log_2K \right)},
\end{eqnarray}
where $\acute{c}$ is a finite constant. The right-hand side vanishes in the limit $K\rightarrow \infty$ provided that $N$ goes to infinity in such a way as to satisfy 
\begin{equation} \label{54aaa}
\zeta=\frac{K}{N\log_3K}<\frac{{\rm{rank}}(\Sigma)}{2}.
\end{equation}
 This proves that $S_7\rightarrow 0$ holds when $\zeta<\frac{{\rm{rank}}(\Sigma)}{2}$. 

One can easily apply the Borel-Cantelli lemma and conclude that a matrix $\matr{H}$ fulfilling the conditions in T\textsc{heorem}~\ref{theorem4} is detecting almost surely.
\end{proof}
 As an example, if $H_{m,n}\in\left\{\pm \frac{1}{\sqrt{2N}}\pm \frac{{\rm j}}{\sqrt{2N}}\right\}$, where ${\rm j}=\sqrt{-1}$, \eqref{54aaa} is written as $\zeta<1$. T\textsc{heorem}~\ref{theorem4} is a generalized form of the theorem presented in \cite{erdgs1963two}. Note that T\textsc{heorem}~\ref{theorem4} only shows that a spreading matrix is a detecting matrix if $\zeta<\frac{{\rm{rank}}(\Sigma)}{2}$. However, as aforementioned the optimum asymptotic multiuser efficiency may vanish for a detecting matrix in the large system limit. 
For the matrix whose elements are chosen randomly from a finite set as defined in T\textsc{heorem}~\ref{theorem4}, the optimum asymptotic multiuser efficiency is not known when the loading factor grows logarithmically with $K$. However, one can consider the condition $\zeta<\frac{{\rm{rank}}(\Sigma)}{2}$ as a necessary condition to have $\eta\rightarrow 1$.

\section{Conclusion}
We derived a bound for the optimum asymptotic multiuser efficiency of a randomly spread CDMA with binary antipodal and Gaussian spreading. BPSK input signals were considered. In fact, we obtained a condition such that the performances of the optimum detector in binary antipodal and Gaussian randomly spread CDMA systems are close to the single user performance at high SNR when $K$ and $N\rightarrow \infty$ and even $\frac{K}{N}$ grows logarithmically with $K$. The connection between detecting matrices and the optimum asymptotic multiuser efficiency was also investigated. It was proven that for any binary input, an $N\times K$ random matrix whose entries are chosen randomly from a finite set, $\Sigma$, with a symmetric pdf is a detecting matrix if $K$ and $N\rightarrow \infty$ and $\frac{K}{N\log_3K}$ is kept less than $\frac{{\rm{rank}}(\Sigma)}{2}$.

\begin{appendices}
\section{}\label{A}
L\textsc{emma}~\ref{Lemma 3} is proven as follows:
\begin{proof}
 Assume that $\matr{x}_m$ is a deterministic vector with weight $m$ and each element of the matrix $\matr{H}$ has a symmetric pdf $ \rho(\cdot)$. To prove the lemma it is enough to show that the pdf of $\matr{x}_m^T\matr{R}\matr{x}_m$ only depends on $m$. Since the elements of $\matr{H}$ are i.i.d., it is concluded that $u_\ell(\matr{x}_m)$ defined in \eqref{1} for $\ell=1\cdots N$ are independent and have the same pdf. Therefore, from \eqref{1} it is enough to show that the pdf of $u_\ell(\matr{x}_m)$ only depends on $m$. From the definition, $u_\ell(\matr{x}_m)=\sum_{n=1}^{K}H_{\ell,n}x_{m}(n)$. This means that $u_\ell(\matr{x}_m)$ is equal to a linear combination of $m$ independent random variables. Furthermore, since $\rho (\cdot)$ is symmetric, the pdf of $H_{\ell,n}x_{m}(n)$ is equal to the pdf of $H_{\ell,n}$ for $x_{m}(n)\in \{\pm 1\}$. Therefore, it is concluded that the pdf of $u_\ell(\matr{x}_m)$ is equal to 
\begin{equation}\label{wightdep}
\rho(H_{\ell,i_1})\ast \rho(H_{\ell,i_2})\ast \cdots \ast\rho(H_{\ell,i_m}),
\end{equation}
where $\ast$ is the convolution operator and $\{i_1,\cdots,i_m\}$ are the indices of those elements of $\matr{x}$ which are not zero. From \eqref{wightdep} it is observed that the pdf of $u_\ell(\matr{x}_m)$ only depends on the weight $m$ and this proves the lemma. 
\end{proof}
\section{}\label{B}
From \eqref{11},
\begin{equation}
\frac{{\rm{p}}(j+1)}{{\rm{p}}(j)}=\frac{2j+1}{2j+2},
\end{equation}
which implies that $\frac{{\rm{p}}(j+1)}{{\rm{p}}(j)}<1$ for $j\geq 1$. Therefore, the function ${\rm{p}}(j)$ is a decreasing function and 
\begin{equation}
{\rm{p}}(j)\leq {\rm{p}}(1)=\frac{1}{2} \Rightarrow 1-p(j)\geq \frac{1}{2},
\end{equation}
and therefore
\begin{equation}
 \left \lfloor N(1-{\rm{p}}(j))\right \rfloor \geq \left \lfloor \frac{N}{2}\right \rfloor.
\end{equation}

\section{}\label{C}
Let $r\in\{1,\cdots,m\}$. By using 
\begin{eqnarray}\label{ApB1}
1=	\sum_{n=0}^{m} {m \choose n} q^n (1-q)^{m-n}\geq {m \choose r} q^r (1-q)^{m-r},
\end{eqnarray}
and letting $q=\frac{r}{m}$, it is obtained that
\begin{eqnarray}\label{ApB2}
 {m \choose r} 2^{-m{\rm{h}}\left( \frac{r}{m}\right)}\leq 1.
\end{eqnarray}
Thus, \eqref{14A} is proven.

\section{}\label{D}
Based on the main theorem in \cite{robbins1955remark} on Stirling's formula, we can write
\begin{eqnarray}\label{str}
 \sqrt{2\pi}n^{n+\frac{1}{2}}{\rm e}^{-n}\leq n!\leq {\rm e} n^{n+\frac{1}{2}}{\rm e}^{-n}.
\end{eqnarray}
By using the bounds \eqref{str}, we obtain
\begin{eqnarray}\label{str2}
 p(j)=\frac{(2j)!}{(j!)^2}2^{-2j}\leq \frac{ {\rm e} (2j)^{2j+\frac{1}{2}}{\rm e}^{-2j}}{2\pi j^{2j+1}{\rm e}^{-2j}}2^{-2j}=\frac{{\rm e}}{\pi \sqrt{2j}}.
\end{eqnarray}

\end{appendices}

\section*{Acknowledgment}
The authors would like to thank Dr. Kasra Alishahi for helpful comments, and the associate editor and the anonymous reviewers for great comments and suggestions.

\bibliographystyle{IEEEtran}
\bibliography{lit}  

\end{document}